%% file: recmorphing.tex
\documentclass[envcountsame,runningheads]{llncs}
\input{config}

\title{Morphing Rectangular Duals} 

\author{Steven~Chaplick\inst{1}\orcidID{0000-0003-3501-4608}
\and Philipp~Kindermann\inst{2}\orcidID{0000-0001-5764-7719}
\and Jonathan~Klawitter\inst{3,4}\orcidID{0000-0001-8917-5269}
\and Ignaz~Rutter\inst{5}\orcidID{0000-0002-3794-4406}
\and Alexander~Wolff\inst{3}\orcidID{0000-0001-5872-718X}}
\authorrunning{Chaplick et al.}
\institute{Maastricht University, The Netherlands
\and Universität Trier, Germany
\and Universität Würzburg, Germany
\and University of Auckland, New Zealand
\and Universität Passau, Germany
}

\begin{document}

\maketitle        

\pdfbookmark[1]{Abstract}{Abstract}
\begin{abstract}
  A rectangular dual of a plane graph $G$ is a contact
  representations of $G$ by interior-disjoint axis-aligned rectangles
  such that (i)~no four rectangles share a point and (ii)~the union of
  all rectangles is a rectangle.  A rectangular dual gives rise to a
  regular edge labeling (REL), which captures the orientations of the
  rectangle contacts.

  We study the problem of morphing between two
  rectangular duals of the same plane graph.  If we require that, at
  any time throughout the morph, there is a rectangular dual, then a
  morph exists only if the two rectangular duals realize the same REL.
  Therefore, we allow intermediate contact representations of
  non-rectangular polygons of constant complexity.  Given an
  $n$-vertex plane graph, we show how to compute in $\Oh(n^3)$ time a
  piecewise linear morph that consists of $\Oh(n^2)$ linear morphing
  steps.
  
  \keywords{morphing \and rectangular dual \and regular edge labeling
    \and lattice}
\end{abstract}

\section{Introduction} 
\label{sec:intro}
A \emph{morph} between two representations (e.g., drawings) of the same graph~$G$ 
is a continuous transformation from one representation to the other.
Preferably, a morph should preserve the user's ``mental map'', 
which means that, throughout the transformation, as little as necessary is changed to go from  
the source to target representation and that their properties are maintained~\cite{PHG07}.
For example, during a morph between two planar drawings,
each intermediate drawing should also be planar.
A \emph{linear morph} moves each point along a straight-line segment
at constant speed, where different points may have different speeds or
may remain stationary.  Note that a linear morph is fully defined by
the source and target representation.  A \emph{piecewise linear morph}
consists of a sequence of linear morphs, each of which is called a
\emph{step}.  

Morphs are well studied for planar drawings.
For example, it is known that piecewise linear planar morphs always exist 
between planar straight-line drawings~\cite{Cai44} 
and that, for an $n$-vertex planar graph, $\Oh(n)$ steps
suffice~\cite{AAB+17}, which is worst-case optimal.
Further research on morphs includes, among others, 
the study of morphs of convex drawings~\cite{ADFLPR15,KK85},
of orthogonal drawings~\cite{BLPS13,GSV19}, 
on different surfaces~\cite{KL08,CELP21}, and in higher
dimensions~\cite{ABP+19}.

Less attention has been given to morphs of alternative representations of graphs 
such as intersection and contact representations. 
A \emph{geometric intersection representation} of a graph~$G$ is a mapping~$\R$ 
that assigns to each vertex~$w$ of~$G$ a geometric object~$\R(w)$ 
such that two vertices~$u$ and~$v$ are adjacent in~$G$ 
if and only if~$\R(u)$ and~$\R(v)$ intersect.  
In a \emph{contact representation} we further require that,
for any two vertices~$u$ and~$v$, the objects~$\R(u)$ and~$\R(v)$ have disjoint interiors.
Classic examples are interval graphs~\cite{BL76}, 
where the objects are intervals of~$\mathbb{R}$, or coin graphs~\cite{Koe36}, 
where the objects are interior-disjoint disks in the plane.
Recently, Angelini \etal~\cite{ACCDR19} studied morphs of 
right-triangle contact representations of planar graphs.
They showed that one can test efficiently whether a morph exists
(in which case a quadratic number of steps suffice).
In this paper, we investigate morphs between contact representations
of rectangles.

\vspace*{-1.5ex}

\paragraph{Rectangular duals.}
A \emph{rectangular dual} of a graph~$G$ is a contact representation~$\R$ of~$G$ 
by axis-aligned rectangles such that (i)~no four rectangles share a point 
and (ii)~the union of all rectangles is a rectangle; see \cref{fig:dual}. 
Note that~$G$ may admit a rectangular dual only if it is planar and internally triangulated.
Furthermore, a rectangular dual can always be augmented with four additional
rectangles (one on each side) so that only these four rectangles touch the
outer face of the representation.
It is customary that the four corresponding vertices on the outer face of $G$ are denoted 
by $\vSouth$, $\vWest$, $\vNorth$, and $\vEast$, 
and to require that $\R(\vSouth)$ is bottommost, $\R(\vWest)$ is leftmost, 
$\R(\vNorth)$ is topmost, and $\R(\vEast)$ is rightmost; see \cref{fig:dual}.
The corresponding vertices are \emph{outer}; the remaining
ones are \emph{inner}.
Similarly, the four edges between the outer vertices are \emph{outer}; the
others are \emph{inner}.
A plane internally-triangulated graph has a representation 
with only four rectangles touching the outer face
if and only if its outer face is a 4-cycle and it has no
\emph{separating triangle},
that is, a triangle whose removal disconnects the
graph~\cite{KK85}. 
Such a graph is called a \emph{properly-triangulated planar (PTP) graph}. 
For such a graph, a rectangular dual can be computed in linear
time~\cite{KH97}.

\begin{figure}[tb]
  \centering
  \includegraphics{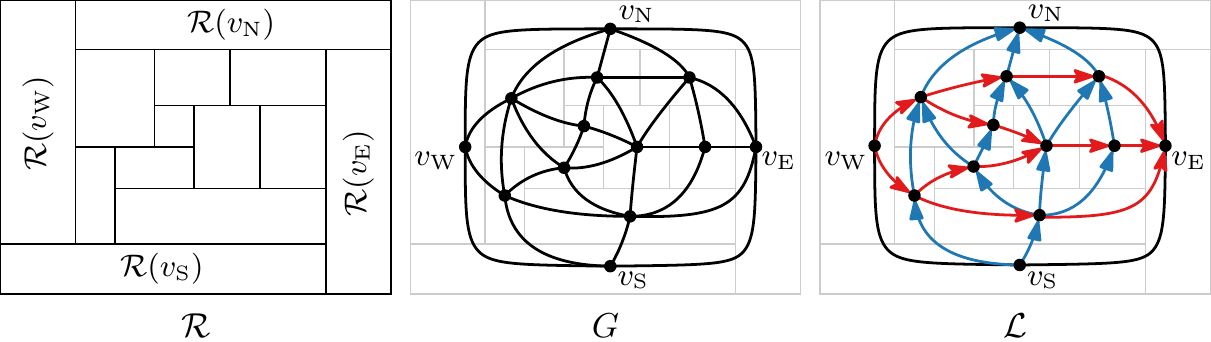}
  \caption{A rectangular dual $\R$ for the graph $G$; the REL $\cL$ induced by~$\R$.}
  \label{fig:dual}
\end{figure}

Historically, rectangular duals have been studied due to their applications in architecture~\cite{Ste73}, 
VLSI floor-planning~\cite{LL84,YS95}, and cartography~\cite{GS69}.
Morphs between rectangular duals are of interest, e.g., due to their relation to rectangular cartograms. 
Rectangular cartograms were introduced in 1934~\cite{Raisz34} and
combine statistical and geographical information in thematic maps, 
where geographic regions are represented as rectangles and scaled in proportion to some statistic.
There has been a lot of work on efficiently computing rectangular cartograms~\cite{HKPS04,vKS07,BSV12}, see also the recent survey~\cite{NK16}.
A morph between rectangular cartograms can visualize different data sets.
Florisson \etal~\cite{FKS05} implemented a method to
construct rectangular cartograms by first extending the given map
with ``sea tiles'' to obtain a rectangular dual,
and then using a heuristic that moves maximal line segments until
the area of the rectangles gets closer to the given data.
They also used their heuristic to morph between two rectangular cartograms,
but did not discuss when exactly this works and with what time complexity.

\vspace*{-1.5ex}

\paragraph{Regular edge labelings.}
A combinatorial view of a rectangular dual of a graph~$G$ 
can be described by a coloring and orientation of the edges of~$G$~\cite{KH97}. 
This is similar to how so-called Schnyder woods describe contact representations
of planar graphs by triangles~\cite{dFdMR94}.
More precisely, a rectangular dual $\R$ gives rise to a 2-coloring 
and an orientation of the inner edges of $G$ as follows.  
We color an edge~$\set{u, v}$ blue if the contact segment between $\R(u)$ 
and $\R(v)$ is a horizontal line segment, and we color it red otherwise.  
We orient a blue (red) edge $\set{u, v}$ as $uv$ if $\R(u)$ lies below (resp. left of) $\R(v)$; see \cref{fig:dual}. 
The resulting coloring and orientation has the following properties~(\cref{figs:REL:edgeOrder}):
\begin{enumerate}[label=(\arabic*)]
\item For each outer vertex $\vSouth$, $\vWest$, $\vNorth$, and
  $\vEast$, the incident inner edges are blue outgoing, red outgoing,
  blue incoming, and red incoming, respectively.
\item For each inner vertex, the incident edges form four clockwise
  (\emph{cw}) ordered non-empty blocks: blue incoming, red incoming,
  blue outgoing, 
  red outgoing.
\end{enumerate}

\begin{figure}[tb]
	\centering
	\includegraphics{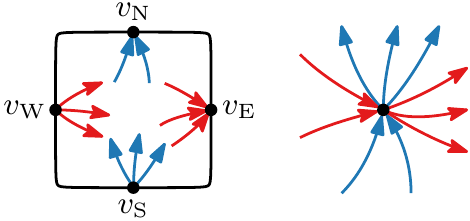}
	\caption{Edge order at the four outer vertices and at an inner
          vertex in a REL.}
	\label{figs:REL:edgeOrder}
\end{figure}

A coloring and orientation with these properties is called a
\emph{regular edge labeling (REL)} or \emph{transversal structure}.
We let $\cL = (L_1,L_2)$ denote a REL,
where $L_1$ is the set of blue edges and $L_2$ is the set of red edges.
Let~$L_1(G)$ and~$L_2(G)$ denote the two subgraphs of~$G$ induced by~$L_1$
and~$L_2$, respectively.
Note that both~$L_1(G)$ and $L_2(G)$ are st-graphs, that is,
directed acyclic graphs with exactly one source and exactly one sink. 
Kant and He~\cite{KH97} introduced RELs as intermediate objects
when constructing a rectangular dual of a PTP graph.
It is well known that every PTP graph admits a REL and thus a rectangular dual~\cite{He93,KH97}. 
A rectangular dual $\R$ \emph{realizes} a REL $\cL$ if the
REL induced by~$\R$ is $\cL$.

We define the \emph{interior} of a cycle to be the set of vertices and
edges enclosed by, but not on the cycle.    A 4-cycle is \emph{separating} if there are other
vertices both in its interior and in its exterior.
A separating 4-cycle is \emph{nontrivial} if its interior contains more than one vertex;
otherwise it is \emph{trivial}.  
We call non-separating 4-cycles also \emph{empty 4-cycles}.
(An empty 4-cycle contains exactly one edge.)

If the edges of a cycle~$C$ alternate between red
and blue, we say that~$C$ is \emph{alternating}.
We can move between different RELs of a PTP graph $G$ 
by swapping the colors and reorienting the edges inside an alternating 4-cycle, 
see \cref{fig:RELrotations}.
\begin{figure}[tb]
  \centering
    	\begin{subfigure}[t]{.47 \linewidth}
		\centering
		\includegraphics[page=1]{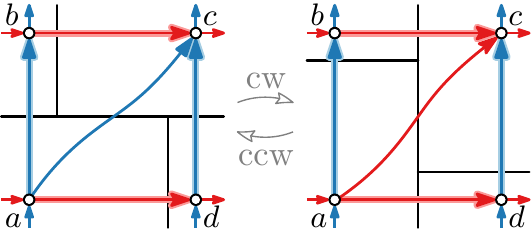}
		\caption{Rotating an empty 4-cycle.}
		\label{fig:RELrotations:empty}
	\end{subfigure}
	\hfill
	\begin{subfigure}[t]{.47 \linewidth}
		\centering
		\includegraphics[page=2]{RELrotations}
		\caption{Rotating a separating 4-cycle.}
		\label{fig:RELrotations:separating}
	\end{subfigure}
  \caption{Clockwise and counterclockwise rotations between RELs that recolor and reorient the edges
    inside an alternating 4-cycle $\croc{a, b, c, d}$ that is
    (a)~empty or (b)~separating.}
  \label{fig:RELrotations}
\end{figure}
This operation, which we call a \emph{rotation} and 
which we define formally in \cref{app:lattice}, 
connects all RELs of $G$.  
In fact, the RELs of~$G$ form a distributive lattice~\cite{Fus06,Fus09}.
A 4-cycle~$C$ of~$G$ is called \emph{rotatable}  
if it is alternating for at least one REL of~$G$.

\vspace*{-1.5ex}

\paragraph{Important related work.}
Other combinatorial structures of graph representations also form lattices; 
see the work by Felsner and colleagues~\cite{FZ07,Fel13}.
In the context of morphs,
Barrera-Cruz \etal~\cite{BHL19} exploited 
the lattice structure of Schnyder woods of a plane triangulation
to obtain piecewise linear morphs between planar straight-line drawings.
While their morphs require $\Oh(n^2)$ steps (compared to the optimum of $\Oh(n)$),
they have the advantage that they are ``visually pleasing''
and that they maintain a quadratic-size drawing area between any two
steps. To this end, Barrera-Cruz \etal showed that there is a path in
the lattice of length $\Oh(n^2)$ between any two Schnyder woods.
We show an analogous result for RELs.
In order to morph between right-triangle contact representations, 
Angelini \etal~\cite{ACCDR19} leveraged
the lattice structure of Schnyder woods.  They showed that if no
separating triangle has to be flipped (a flip is a step in the
lattice) between the source and the target Schnyder wood, then a morph
with $\Oh(n^2)$ steps exists (else, no morph exists that uses
right-triangle contact representations throughout).

\vspace*{-1.5ex}

\paragraph{Contribution.}
We consider piecewise linear morphs between two rectangular duals 
$\R$ and~$\R'$ of the same PTP graph~$G$.
If $\R$ and $\R'$ realize the same REL, then a single step suffices,
but if $\R$ and $\R'$ realize distinct RELs of $G$, 
then no rectangular-dual preserving morph exists (\cref{sec:sameREL}).
Therefore, we propose a new type of morph where intermediate drawings
are contact representations of~$G$ using convex polygons with up to
five corners (\cref{sec:singleRelaxed}). 
We show how to construct such a \emph{relaxed} morph
as a sequence of $\Oh(n^2)$ steps that implement moves in the lattice
of RELs of~$G$ (\cref{sec:serial}).
To this end, we make use of the following 
two results on paths in this lattice (which we prove in \cref{app:lattice}).

\begin{restatable}[{\hyperref[clm:RELdistance*]{\appmark}}]{proposition}{RELdistance}
  \label{clm:RELdistance}
  Given an $n$-vertex PTP graph $G$, the lattice of RELs of $G$ has
  diameter~$\Oh(n^2)$.
\end{restatable}

\begin{restatable}[{\hyperref[clm:RELshortestPath*]{\appmark}}]{proposition}{RELshortestPath}
  \label{clm:RELshortestPath}
  Let $G$ be an $n$-vertex PTP graph with RELs $\cL$ and $\cL'$.  In
  the lattice of RELs of $G$, a shortest $\cL$--$\cL'$ path can be
  computed in $\Oh(n^3)$ time.
\end{restatable}

We ensure that between any two morphing steps, our drawings remain on
a quadratic-size section of the
integer grid~-- like those of Barrera-Cruz \etal~\cite{BHL19}.
In order to evaluate the intermediate representations when our
drawings are not on the integer grid, we use the measure \emph{feature
  resolution}~\cite{HKKR14}, that is, the ratio of the length of the
longest segment over the shortest distance between two vertices or
between a vertex and a non-incident segment.  
We show that the feature resolution in {\em any} intermediate drawing is bounded by $\Oh(n)$.

Finally, we investigate executing rotations in parallel; see \cref{sec:parallel}.  
As a result,
we can morph between any pair $(\R,\R')$ of rectangular duals of the
given graph using $\Oh(1)$ times the minimum number of steps
needed to get from~\R to~$\R'$; however, our polygons have up to eight corners.

We prove every statement marked with a (clickable) ``\appmark'' in the appendix.

\section{Morphing between Rectangular Duals} 
\label{sec:rotation}

This section concerns morphs between two rectangular duals $\R$ and~$\R'$ of a PTP graph~$G$
that realize the same REL, adjacent RELs, and finally any two RELs.

\subsection{Morphing Between Rectangular Duals Realizing the Same REL} 
\label{sec:sameREL}

\begin{theorem} \label{clm:sameREL}
For a PTP graph $G$ with rectangular duals $\R$ and $\R'$,
(i) if $\R$ and $\R'$ realize the same REL, then there is a linear morph between them;
(ii) otherwise, there is no morph between them (not even a non-linear one).
\end{theorem}
\begin{proof}
Biedl \etal~\cite{BLPS13} studied morphs of orthogonal drawings. 
They showed that a single (planarity-preserving) linear morph suffices 
if all faces are rectangular and all edges are parallel in the two drawings,
that is, any edge is either vertical in both drawings or horizontal in both drawings.  
We can apply this result to two rectangular duals $\R$ and $\R'$ 
precisely when they realize them same REL. 
A linear morph between them changes the x-coordinates of vertical line segments
and the y-coordinates of horizontal line segments but does not change their relative order.

Now assume that $\R$ and $\R'$ realize different RELs $\cL$ and $\cL'$ of $G$, respectively.
Then w.l.o.g.\ some contact segment $s$
changes from being horizontal in~$\R$ to being vertical in $\R'$.
Since~$s$ must always be horizontal or vertical, it has to
collapse to a point and then extend to a segment again.
When $s$ collapses, the intermediate representation is not a
rectangular dual of $G$ since four rectangles meet at a single point.
Even worse, if a separating alternating cycle is rotated, then 
its interior contracts to a point, vanishes, and reappears
rotated by $90^\circ$.
\end{proof}

\subsection{Morphing Between Rectangular Duals with Adjacent RELs} 
\label{sec:singleRelaxed}
Let $\R$ and $\R'$ now realize different RELs~$\cL$ and $\cL'$ of $G$, respectively. 
By \cref{clm:sameREL}, any continuous transformation between $\R$ and~$\R'$ requires intermediate representations 
that are not rectangular duals of $G$, i.e., a morph in the traditional sense is not possible.
We relax the conditions on a morph such that, in an intermediate contact representation of $G$, 
vertices can be represented by convex polygons of constant complexity -- in this section, by 5-gons.
However, we still require that these polygons form
a tiling of the bounding rectangle of the representation.
We call a transformation with this property a \emph{relaxed morph}.
(When we talk about linear morphing steps, we omit the adjective ``relaxed''.)
The following statement describes relaxed morphs when $\cL$ and $\cL'$ are adjacent, that is,
$(\cL', \cL)$ is an edge in the lattice of RELs of $G$.

\begin{proposition} \label{clm:rotationalMorph}
Let $\R$ and $\R'$ be two rectangular duals of an $n$-vertex PTP graph $G$
realizing two adjacent RELs $\cL$ and $\cL'$ of $G$, respectively.
Then, we can compute in $\Oh(n)$ time a 3-step relaxed morph
between~$\R$ and~$\R'$.
If $\R$ and $\R'$ have an area of at most $n \times n$ and feature resolution
in $\Oh(n)$, then so has each representation throughout the morph.
\end{proposition}

We assume w.l.o.g.\ that $\cL'$ can be obtained from
$\cL$ by a cw rotation of an alternating 4-cycle $C$.  
The idea is to rotate the interior of the $4$-cycle 
while simultaneously moving the contact segments that form the edges of~$C$;
see \cref{fig:rotationalMorph,fig:rotationalMorphEmpty}.
To ensure that, except for the vertices of~$C$ all regions remain
rectangles and that moving the contact segments of~$C$ does not change
any adjacencies, the representation needs to satisfy certain requirements.  
Therefore, our relaxed morph from $\R$ to $\R'$ consists of three steps.  
First a \emph{preparatory morph} from~$\R$ to a rectangular dual~$\R_1$ with REL~$\cL$ 
for which~$C$ satisfies conditions stated below;
second a \emph{main morph} which transforms~$\R_1$
to a rectangular dual~$\R_2$ whose REL is~$\cL'$, and third a
\emph{clean-up morph} that transforms~$\R_2$ into $\R'$.

\begin{figure}[tb]
  \centering
  \includegraphics{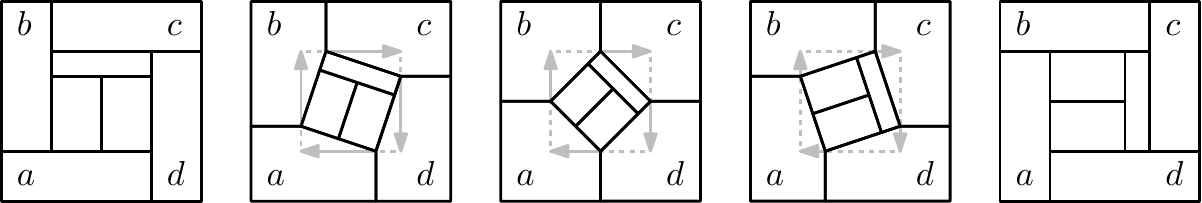}
  \caption{A linear morph rotating the separating 4-cycle
    $\croc{a, b, c, d}$ emulates the rotation in the
    corresponding REL; see \cref{fig:RELrotations}. The interior of
    the 4-cycle turns by 90$^{\circ}$ without changing its shape,
    while the outer contact segments move horizontally and vertically.}
  \label{fig:rotationalMorph}
\end{figure}

\begin{figure}[tb]
  \centering
  \includegraphics{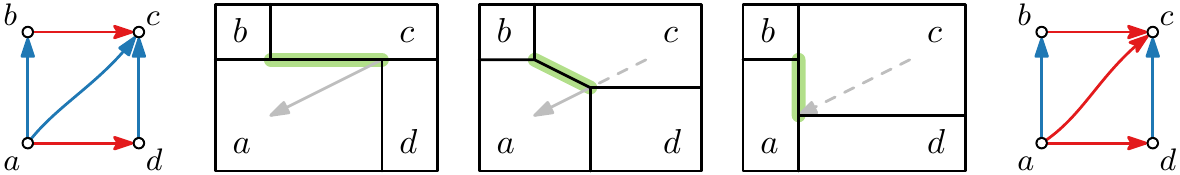}
  \caption{A linear morph that rotates the inner contact segment of an empty 4-cycle.}
  \label{fig:rotationalMorphEmpty}
\end{figure}

We first describe the main morph $\croc{\R_1, \R_2}$ in detail, as
this allows us to also infer the conditions under which it can be
executed successfully. Then we describe the preparatory~morph
$\croc{\R, \R_1}$, whose sole purpose is to ensure the conditions that
are required for the main morph.

\vspace*{-1.5ex}

\pdfbookmark[3]{Main Morph}{Main Morph}
\paragraph{Main morph $\croc{\R_1, \R_2}$ to rotate $C$.}  
Let $a$, $b$, $c$, and $d$ be the vertices of $C$ in cw order where $a$
is the vertex with an outgoing red and outgoing blue edge in $C$,
i.e., it corresponds to the bottom-left rectangle of~$C$.  

Assume for now that~$C$ is separating.
We have the following requirements for~$R_1$, which become apparent shortly.
\begin{quote}
(P1) The rectangle~$I_C$ bounding the interior of~$C$ is a square.
\end{quote}
Next, we consider the four maximal segments of $\R$ that contain one of the
four borders of $I_C$, which we call \emph{border segments}.
Let $s$ be the upper border segment of $I_C$
and suppose its right endpoint lies on the left side of a rectangle $\R(x)$.
Let $S$ be the part in the horizontal strip defined by~$I_C$ 
that starts at $I_C$ and ends at~$\R(x)$.
\begin{quote}
(P2) The only horizontal segments that intersect $S$ are border segments of $I_C$; see \cref{fig:prepMorph:result}.
\end{quote}
We define (P2) for the other three border segments of $I_C$ analogously. 
Next, assume that~$C$ is empty. 
Then the rectangle~$I_C$ degenerates to a segment~$s$, 
and we assume w.l.o.g.\ that~$s$ is horizontal. 
Now $I_C$ still has two vertical border segments, 
but the two horizontal border segments share the segment~$s$.
Let $s$ have again its right endpoint on the left side of $\R(x)$.
Let $S$ be the rectangular area of height~1 directly below $s$
that starts at $b$ and ends at $\R(x)$. We have the following requirement if $C$ is empty.
\begin{quote}
(P2') The only horizontal segment that intersects $S$ is $s$; see \cref{fig:prepMorph:empty:result}.
\end{quote}
There is no requirement for the left side of $s$ and the left vertical border segment of $I_C$.
The requirements for the right vertical border segment is (P2).

We now describe the main morph for the case that~$C$ is a separating 4-cycle.
In this case, the interior of~$C$ forms a square in~$\R_1$ by (P1).  
Recall that the rotation of $C$ from $\cL$ to $\cL'$ turns the interior of $C$ by~$90^{\circ}$.
During the morph, we move each corner of~$I_C$ to the coordinates of
the corner that follows in cw order around $I_C$ in $\R_1$; see \cref{fig:rotationalMorph}.
All other points in $I_C$ are expressed as a convex combination of the corners of~$I_C$ 
and then move according to the movement of the corners.  
Furthermore, we move all points on the left
border segment of $I_C$ that are outside the boundary of~$I_C$
horizontally to the x-coordinate of the right side of $I_C$. We move
the points on other border segments of $I_C$ analogously.

This describes a single linear morph that results in a rectangular dual~$R_2$
that realizes the REL $\cL'$. Since $I_C$ starts out as a square by (P1), 
throughout the morph, $I_C$ remains a square,
and by similarity all rectangles inside $I_C$ remain rectangles. 
The rectangles $a$, $b$, $c$, and $d$ become convex 5-gons.
Furthermore, since outside $I_C$ the horizontal border segments of $I_C$ 
move an area that contains no other horizontal segments by (P2), no contact along a vertical segment arises or vanishes. 
Analogously, for the vertical border segments, no contact along a horizontal segment arises or vanishes. 
Hence, we maintain the same adjacencies.

Note that, if $I_C$ would not be a square, then its corners would move at 
different speeds and $I_C$ would deform to a rhombus where the inner angles 
are not $90^\circ$, and so would all the rectangles inside $I_C$.

Next, consider the case that $C$ is an empty 4-cycle.
Recall that in this case, the rectangle~$I_C$ degenerates to a
segment~$s$ and we assume that $s$ is horizontal.
Note that~$\R'$ has a vertical contact between $a$ and $c$,
since we assume a cw rotation from $\cL$ to $\cL'$.
We then we move the right endpoint of $s$ vertically down by 1 and horizontally to the 
x-coordinate of the left endpoint of $s$;
see~\cref{fig:rotationalMorphEmpty}. 
We also move all points on the border segments that contain the
right endpoint of $s$ accordingly.
The rectangles $a$ and $c$ become convex 5-gons.
Furthermore, since outside $I_C$ only the horizontal border segments of $I_C$ 
lie inside the area of height~1 below $s$ by (P2'), 
no contact along a vertical segment arises or vanishes. 
Analogously, due to condition (P2) for the vertical border segments, no contact along a 
horizontal segment arises or vanishes. Hence,  maintaining the same adjacencies.

To show that the feature resolution remains in $\Oh(n)$, note that both $\R_1$ and $\R_2$ are drawn 
on a $n\times n$ grid. 
Furthermore, the rectangles inside $I_C$ are scaled
during the morph, but since $I_C$ is a square, the whole area inside $I_C$
is scaled by at most $\sqrt{2}$. The distances outside $I_C$ cannot become
smaller than $1$.

\begin{lemma}\label{clm:mainmorph}
	Let $\R_1$ and $\R_2$ be two rectangular duals of an $n$-vertex PTP graph~$G$
	realizing two adjacent RELs $\cL$ and $\cL'$ of $G$, respectively,
	such that $\R_1$ satisfies (P1) and (P2) (or (P1) and (P2')).
	Then, we can compute in $\Oh(n)$ time a relaxed morph between~$\R_1$ and~$\R_2$.
	If~$\R_1$ and~$\R_2$ have an area of at most $n \times n$ and feature 
	resolution in $\Oh(n)$, then so has each representation throughout the~morph.
\end{lemma}

\vspace*{-1.5ex}

\pdfbookmark[3]{Preparatory Morph}{Preparatory Morph}
\paragraph{Preparatory morph $\croc{\R, \R_1}$.}
We consider again the case where $C$ is separating first.
To obtain $\R_1$ from $\R$,
we extend $G$ to an auxiliary graph $\hat G$ 
that is almost a PTP graph but that contains empty chordless 4-cycles
(which are represented by four rectangles touching in a single point).
For $\hat G$, we compute an auxiliary REL $\hat \cL$ 
where the empty chordless 4-cycles of $\hat G$ are colored alternatingly.
We then use the second step of the linear-time algorithm by Kant and He~\cite{KH97} to compute
an (almost) rectangular dual $\hat \R_1$ of $\hat G$ that realizes $\hat \cL$.
By reversing the changes applied to $G$ to obtain $\hat G$, we derive $\R_1$ from $\hat \R_1$.
We explain in \cref{app:kanthe} the algorithm by Kant and He and why it also works for $\hat G$. 

We start with the changes to ensure (P2) for the upper border segment~$s$ of~$I_C$; 
it works analogously for the other border segments.
Let~$s$ end to the right again at~$\R(x)$.
Let~$P_d$ be the leftmost path in~$L_2(G)$ from~$d$ to~$x$.
Let $\y_1(R)$ and $\y_2(R)$ denote the lower and upper y-coordinate of a rectangle $R$, respectively.
Note that (P2) holds if for each vertex~$v$ on~$P_d$ we have~$\y_1(\R_1(v)) < \y_2(\R_1(a))$. 
Therefore, from~$G$ to~$\hat G$,
we duplicate~$P_d$ by splitting each vertex~$v$ on~$P_d \setminus \set{x}$ into two vertices~$v_1$ and~$v_2$; see \cref{fig:prepMorph}.
We then connect~$v_1$ and~$v_2$ with a blue edge.
Let~$y$ be the successor of~$d$ on~$P_d$.
We assign the edges cw between (and including)~$dy$ and~$ad$ to~$d_1$, 
and the edges cw between~$ad$ and~$dy$ to~$d_2$.
If~$x = y$, the edge~$dy$ is assigned to both~$d_1$ and~$d_2$;
otherwise we replace~$dy$ with~$d_1 y_1$ and~$d_2 y_2$.
For all other vertices~$v$ on~$P_d \setminus \set{d, x}$, 
let~$u$ be the predecessor and let~$w$ be the successor of~$v$ on~$P_d$.
We assign the edges cw between~$vw$ and~$uv$ to~$v_1$, and the edges
cw between~$uv$ and~$vw$ to~$v_2$; furthermore, we add the edges
$u_1v_1$, $v_1w_1$, $u_2v_2$, and~$v_2w_2$.
As a result, there is a path from~$a$ to~$x$ in~$\hat \R_1$ 
through the ``upper'' copies of the vertices in~$P_d \setminus \set{x}$, 
and the bottom side of their corresponding rectangles are aligned.
Hence, $\y_1(\hat\R_1 (v_1)) < \y_1(\hat\R_1(v_2)) = \y_2(\hat\R_1(a))$ for every~$v \in P_d \setminus \set{x}$,
and $\y_1(\hat\R_1(x)) < \y_2(\hat\R_1(a))$.
We obtain for each~$v$ on~$P_d \setminus \set{x}$
the rectangle~$\R_1(v)$ by merging~$\hat \R_1(v_1)$ and~$\hat \R_1(v_2)$.
This works analogously if~$C$ is empty; see~\cref{fig:prepMorph:empty}.

\begin{figure}[tb]
  \centering
  	\begin{subfigure}[t]{.3\linewidth}
		\centering
		\includegraphics[page=1]{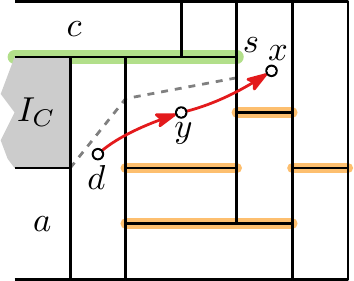}
		\caption{The path $P_d$ in $\R$.}
		\label{fig:prepMorph:path}
	\end{subfigure}
	\hfill
	\begin{subfigure}[t]{.31\linewidth}
		\centering
		\includegraphics[page=2]{prepMorph}
		\caption{Splitting $P_d\setminus\set{x}$,
                    we~obtain~$\hat \R_1$ and~$\hat \cL$.}
		\label{fig:prepMorph:newRecDual}
	\end{subfigure}
	\hfill
	\begin{subfigure}[t]{.31\linewidth}
		\centering
		\includegraphics[page=3]{prepMorph}
		\caption{Deriving $\R_1$ from $\hat \R_1$,
                    we get property~(P2).}
		\label{fig:prepMorph:result}
	\end{subfigure}
  \caption{We compute $\R_1$ via an auxiliary rectangular dual $\hat \R_1$ and an auxiliary REL $\hat \cL$.}
  \label{fig:prepMorph}
\end{figure} 
 
\begin{figure}[tb]
  \centering
  	\begin{subfigure}[t]{.3\linewidth}
		\centering
		\includegraphics[page=1]{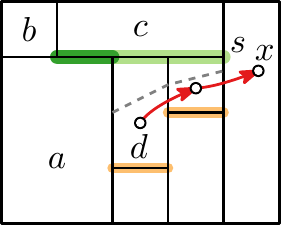}
		\caption{The path $P_d$ in $\R$.}
		\label{fig:prepMorph:empty:path}
	\end{subfigure}
	\hfill
	\begin{subfigure}[t]{.31\linewidth}
		\centering
		\includegraphics[page=2]{prepMorphEmpty}
		\caption{Splitting $P_d \setminus\set{x}$,
                    we~obtain $\hat \R_1$ and $\hat \cL$.}
		\label{fig:prepMorph:empty:newRecDual}
	\end{subfigure}
	\hfill
	\begin{subfigure}[t]{.3\linewidth}
		\centering
		\includegraphics[page=3]{prepMorphEmpty}
		\caption{Deriving $\R_1$ from $\hat \R_1$, we get property (P2').} 
		\label{fig:prepMorph:empty:result}
	\end{subfigure}
  \caption{Preparatory morph analogous \cref{fig:prepMorph} for the case when $C$ is empty.}
  \label{fig:prepMorph:empty}
\end{figure}  

Next, we describe how to ensure (P1) in $\R_1$,
i.e., that the interior $I_C$ of $C$ is a square.
Let $w_C$ and $h_C$ be the minimum width and height, respectively,  
of a rectangular dual of $I_C$. 
These values can be computed in $\Oh(I_C)$ time.
Note that because of (P2), the algorithm by Kant and He~\cite{KH97}
will draw $I_C$ with minimum width and height in $\R_1$:
the algorithm draws every horizontal line segment as low as possible,
and because of (P2) there is no horizontal line segment to the right of $I_C$
that forces the upper boundary segment of $I_C$ to be higher;
a symmetric argument applies to the left boundary segment of $I_C$.
Hence, if $w_C = h_C$, then no further changes to $\hat G$ are required. 
Otherwise, if w.l.o.g.\ $w_C < h_C$, 
we add $h_C - w_C$ many buffer rectangles between $I_C$ and $d$ as follows; see \cref{fig:prepMorph:square}.
Let $\Delta = h_C - w_C$.
From $G$ to $\hat G$, we add vertices $v_1, \ldots, v_\Delta$ with a red path through them
and, for $i \in \set{1, \ldots, \Delta}$, we add the blue edges $av_i$ and $v_ic$.
All incoming red edges of $d$ in $G$ from the interior of $C$
become incoming red edges of $v_1$, and
we add a red edge $(v_\Delta, d_2)$.
In $\hat G$, the minimum width and height of $I_C$ are now the same
and $I_C$ is drawn as a square in $\hat \R_1$ .
To obtain $\R_1$ from $\hat \R_1$, we remove the buffer rectangles
and stretch all right-most rectangles of $I_C$ to $\y_1(\R_1(d))$.  

\begin{figure}[tb]
  \centering
  	\begin{subfigure}[t]{.4 \linewidth}
		\centering
		\includegraphics[page=4]{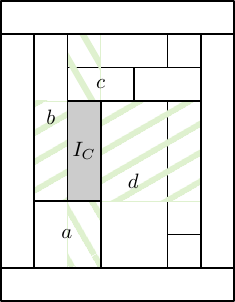}
		\caption{If $I_C$ is not a square in $\R_1$ with (P2), \ldots}
		\label{fig:prepMorph:square:notyet}
	\end{subfigure}
	\hfill
	\begin{subfigure}[t]{.55 \linewidth}
		\centering
		\includegraphics[page=5]{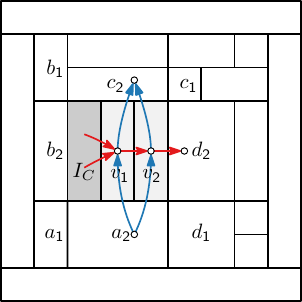}
		\caption{\ldots then we extend $\hat \R_1$ and $\hat \cL$ with dummy vertices $v_i$
		inside $\croc{a, b, c, d}$.}
		\label{fig:prepMorph:square:extend}
	\end{subfigure}
  \caption{To ensure (P1), i.e., that $I_C$ is a square in $\R_1$, we extend $\hat \R_1$ and $\hat \cL$ further.}
  \label{fig:prepMorph:square}
\end{figure}

Concerning the running time, note that we can both find and split the paths for (P2) 
and add the extra vertices for (P1) in $\Oh(n)$ time.
Since $\hat G$ and $\hat \cL$ have a size in $\Oh(n)$,
the algorithm by Kant and He~\cite{KH97} also runs in $\Oh(n)$~time.

Finally, we show that the area of~$\R_1$ is bounded by $n \times n$.
Observe that each triangle in $G$ corresponds to a T-junction in $\R$
and thus to an endpoint of a maximal line segment.
There are $2n-4$ triangles in $G$ and 
thus $n-2$ inner maximal line segments besides the four outer ones.
The algorithm by Kant and He~\cite{KH97} ensures 
that each x- and y-coordinate inside a rectangular dual contains a horizontal or
vertical line segment, respectively.
Note that $\hat \R_1$ contains exactly $\Delta$ more maximal line segments than $\R$.
These were added inside~$C$ if
in~$I_C$ the number of horizontal and vertical maximal line segments
differed by at least~$\Delta$.
Hence, $\R_1$ contains at most $n-2$ vertical and at most $n-2$
horizontal inner maximal line segments.
Thus, the area of~$\R_1$ is bounded by $n \times n$.
Lastly, note that $\hat \R_1$ and $\R_1$ have the same size.
Furthermore, we move points only away from each other, so the feature
resolution remains in $\Oh(n)$.

\begin{lemma} \label{clm:tidyRecDual}
Let $\R$ be a rectangular dual of an $n$-vertex PTP graph $G$
realizing a REL $\cL$ of~$G$.
Let $C$ be an alternating separating 4-cycle in $\cL$.
Then, we can compute in $\Oh(n)$ time a rectangular dual $\R_1$ of $G$ realizing $\cL$ that satisfies the requirements (P1) and (P2).
If $\R$ has an area of at most $n \times n$ and feature resolution $\Oh(n)$,
then so has $\R_1$ and each representation throughout the morph.
\end{lemma}

To prove \cref{clm:tidyRecDual}, we do not use zig-zag moves, which
were introduced for morphing orthogonal drawings~\cite{BLPS13,GSV19},
since then we would not be able to bound the area by $n \times n$
throughout the morph. In order to 
keep a bound~of $\Oh(n) \times \Oh(n)$, it seems that we would need a
re-compactification step after each zig-zag move. Therefore, we keep
the modifications in our morph as local as possible.

Let us now consider the morph $\croc{\R_1, \R_2}$ again. 
Since~$\R_1$ now satisfies (P1) and (P2),
only the inside $I_C$  of $C$, the four rectangles of $C$, and 
the border segments of $I_C$ move.
The target positions of these can be computed in $\Oh(n)$ time. 
The linear morph is then defined fully by the start and target positions.
Furthermore, $\R_2$ and all intermediate representations have the same area as~$\R_1$. 

\begin{proof}[of {\cref{clm:rotationalMorph}}]
  By \cref{clm:sameREL,clm:mainmorph}, we can get from~$\R$ via $\R_1$
  and $\R_2$ to $\R'$ using three steps. 
  The claims on the running time and the area follow from \cref{clm:sameREL},
  \cref{clm:mainmorph,clm:tidyRecDual}, and the observations above.
\end{proof}

\subsection{Morphing Between Rectangular Duals} 
\label{sec:serial}
Combining results from the previous sections, we can now
prove our main result.

\begin{theorem} \label{clm:serialMorphing}
Let $G$ be an $n$-vertex PTP graph with rectangular duals~$\R$ and $\R'$.
We can find in $\Oh(n^3)$~time a relaxed morph between $\R$ and $\R'$ with $\Oh(n^2)$ steps
that executes the minimum number of rotations.
If $\R$ and $\R'$ have an area of at most $n \times n$
and feature resolution in $\Oh(n)$, 
then so does~each representation throughout the morph.
\end{theorem}
\begin{proof}
Let $\cL$ and $\cL'$ be the RELs realized by $\R$ and $\R'$, respectively.
By \cref{clm:RELshortestPath} a shortest path between $\cL$ and $\cL'$
in the lattice of RELs of $G$ can be computed in $\Oh(n^3)$ time, and its length is~$\Oh(n^2)$ by \cref{clm:RELdistance}.
For each rotation along this path, we construct a relaxed morph 
with a constant number of steps in~$\Oh(n)$ time by \cref{clm:rotationalMorph}.
The area and feature resolution also follow from \cref{clm:rotationalMorph}.
\end{proof}

\section{Morphing with Parallel Rotations} 
\label{sec:parallel}
We now show how to reduce the number of morphing steps
by executing rotations in parallel. 
We assume that all separating 4-cycles in our PTP graph $G$
are trivial.

Consider two cw rotatable separating 4-cycles $C$ and $C'$
that share a maximal horizontal line segment $s$ as border segment; see \cref{fig:conflictingRotations}.
If~$C$ contains the left endpoint of $s$,
a rotation of $C$ would move $s$ downwards while a rotation of $C'$ would move $s$ upwards.
Therefore, such a morph skews angles such that they are not multiples of $90^\circ$ 
even at vertices that are not incident to the interior of $C$ or $C'$.
To avoid such morphs, we say that~$C$ and $C'$ are~\emph{conflicting}.
For a set of cw rotatable separating 4-cycles $\cC$ for $\R$,
this gives rise to a \emph{conflict graph} $K(\cC)$ with vertex set $\cC$.
Note that a separating 4-cycle can be in conflict with at most four other separating 4-cycles.
Therefore, $K(\cC)$ has maximum degree four.

\begin{figure}[tb]
  \centering
  \includegraphics{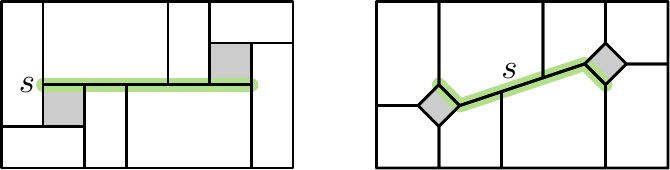}
  \caption{Two conflicting separating 4-cycles that share the interior segment $s$.}
  \label{fig:conflictingRotations}
\end{figure}

Next, consider a separating 4-cycle $C$ that shares a maximal horizontal line segment $s$ 
with an empty 4-cycle $C'$; see \cref{fig:fourRotations8Gon}.
In this case, we can rotate and translate the inner contact segment of $C'$ downwards, 
which allows us to simultaneously rotate~$C$ and $C'$ 
without creating unnecessary skewed angles. 
Also note that two cw rotatable empty 4-cycles may only overlap with one edge
but may not contain an edge of the other. Hence, they are not conflicting.   

\begin{figure}[tb]
  \centering
  \includegraphics{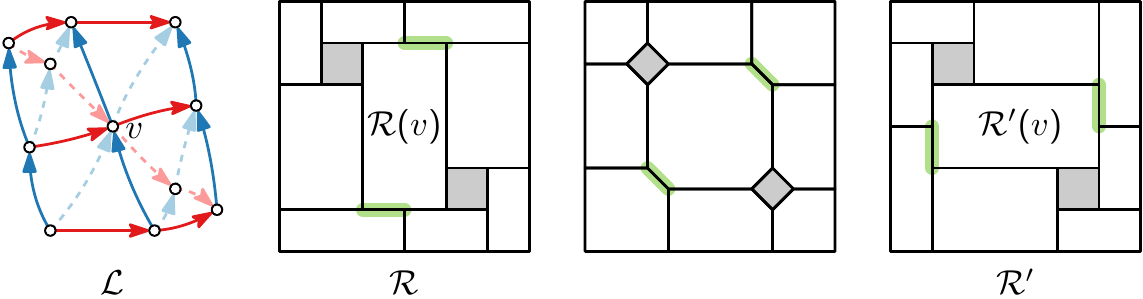}
  \caption{When we rotate four cw alternating 4-cycles that share a vertex $v$ in a single relaxed morph from $\R$ to $\R'$,
  then $v$ is temporarily represented by a convex 8-gon.}
  \label{fig:fourRotations8Gon}
\end{figure}

To rotate a set $\cC$ of alternating 4-cycles using $\Oh(1)$ steps, we
divide $\cC$ into color classes based on $K(\cC)$ and rotate one color
class at a time. 

\begin{restatable}[{\hyperref[clm:simultaneousRotation*]{\appmark}}]{proposition}{simultaneousRotation}
  \label{clm:simultaneousRotation}
  Let $\R$ be a rectangular dual of a PTP graph $G$ with REL $\cL$
  whose separating 4-cycles are all trivial.  Let $\cC$ be a set of
  alternating 4-cycles of $\R$.  Let~$\cL'$ be the REL obtained from
  $\cL$ by executing all rotations in~$\cC$.  There exists a relaxed
  morph with $\Oh(1)$ steps from~$\R$ to a rectangular dual~$\R'$
  realizing~$\cL'$.  The morph can be computed in linear time.
\end{restatable}

Note that there exist rectangular duals with a linear number of
alternating 4-cycles~-- extend \cref{fig:fourRotations8Gon} into a
grid structure.
Hence, parallelization can reduce the number of morphing steps
by a linear factor.
Even more, using \cref{clm:simultaneousRotation}, we obtain the
following approximation result.

\begin{restatable}[{\hyperref[clm:approx*]{\appmark}}]{theorem}{approx}\label{clm:approx}
Let $G$ be a PTP graph whose separating 4-cycles are all trivial.
Let $\R$ and $\R'$ be two rectangular duals of $G$,
and let $\mathrm{OPT}$ be the minimum number of steps
in any relaxed morph between $\R$ and $\R'$.
Then we can construct in cubic time
a relaxed morph consisting of $\Oh(\mathrm{OPT})$ steps.
\end{restatable}

\section{Concluding Remarks} 
In the parallelization step, we considered only PTP graphs whose
separating 4-cycles are trivial.  It remains open how to parallelize
rotations for RELs of PTP graphs with nontrivial separating 4-cycles,
in particular, to construct morphs that execute rotations of nested
4-cycles in parallel.  It would also be interesting to guarantee area
bounds for morphs with parallel rotations. 

During our relaxed morphs, we allow rectangles to temporarily turn
into convex 5-gons (with four edges axis-aligned).  Alternatively, one
could insist that the intermediate objects remain ortho-polygons.
This would require upt to six vertices per shape and would force not only the
outer rectangles in \cref{fig:rotationalMorph} to change their shape,
but also the rectangles in the interior.  We find our approach more
natural.

\pdfbookmark[1]{References}{References}

\bibliographystyle{abbrvurl} 
\bibliography{abbrv,sources}

\clearpage
\pdfbookmark[0]{Appendix}{toc.appendix}
\section*{Appendix}
\appendix
\label{app:appendix}

\section{The Lattice of RELs} 
\label{app:lattice}

In this section, we review, for a given PTP graph, the structure of
its RELs. In particular, we detail how rotations allow us to switch
from one REL to another. We also remind the reader of the lattice of
RELs induced by such rotations.
Moreover, we prove that the lattice has a height of $\Oh(n^2)$ (\cref{clm:RELdistance})
and that a shortest path can be computed in $\Oh(n^3)$ time (\cref{clm:RELshortestPath}).

Let~$G$ be a PTP graph, let~$\cL$ be a REL of~$G$, and let~$C$ be a
4-cycle in~$G$ and~$\cL$. 
Recall that~$C$ is called \emph{alternating}
if the edges of~$C$ alternate between red and blue in~$\cL$.  
By swapping the colors of the edges inside~$C$ and by then fixing the orientations
of the recolored edges (uniquely), we obtain a different REL~$\cL'$
of~$G$; see \cref{fig:RELrotations}.  Considering rectangular duals
that realize~$\cL$ and~$\cL'$, respectively, note that the operation
``rotates'' the interior of the cycle (when~$C$ is
separating)~-- or the interior contact segment~-- by~90$^\circ$.
Hence, we call such an operation either a \emph{clockwise (cw)} or a
\emph{counterclockwise (ccw) rotation}.  A 4-cycle~$C$ of~$G$ is
called \emph{rotatable} if it is alternating for at least one REL
of~$G$.

\paragraph{The lattice.}
Recall that a \emph{lattice} is a poset in which each pair $(a, b)$
has a unique smallest upper bound~-- the \emph{join} $a \vee b$ of $a$
and $b$~-- and a unique largest lower bound~-- the \emph{meet}
$a \wedge b$ of $a$ and $b$.  A lattice is \emph{distributive} if the
join and meet operation are distributive with respect to each other.
Fusy~\cite{Fus06,Fus09} showed that the set of RELs of $G$ forms a
distributive lattice, where $\cL \le \cL'$ if and only if
$\cL'$ can be obtained from $\cL$ by ccw rotations.  
Let $\cH(G)$ denote the lattice of RELs of $G$.
The minimum element $\cL_{\min}$ of $\cH(G)$ is the unique REL that admits no cw
rotation; the maximum element $\cL_{\max}$ is the unique REL that
admits no ccw~rotation.

\paragraph{Shortest paths.}
A path in the REL lattice $\cH(G)$ is \emph{monotone} if it uses only
cw rotations or only ccw rotations.  By Birkhoff's
theorem~\cite{Bir37}, the length of two monotone paths between two
elements in a distributive lattice is equal.  Furthermore, any three
elements $a, b, c$ of a distributive lattice have a unique median
$(a \vee b) \wedge (a \vee c) \wedge (b \vee c) = (a \wedge b) \vee (a
\wedge c) \vee (b \wedge c)$ that lies on a shortest path between any
two of them~\cite{BK47}.  
Applying this to $a$, $b$, and $a \vee b$, yields that $a \vee b$ lies
on a shortest path from~$a$ to~$b$.  Hence, for any pair of RELs,
there exists a shortest path that contains their meet and one that
contains their join.

\paragraph{Upper bound on path lengths.}
Recall that a 4-cycle is separating if there are other vertices both
in its interior and its exterior.  A separating 4-cycle is
\emph{nontrivial} if its interior contains more than one vertex;
otherwise it is \emph{trivial}.  Note that the four vertices of a
separating 4-cycle $C$ can be seen as the outer cycle of a PTP
subgraph and thus, by the coloring rules of RELs, each vertex of $C$
has edges in only one color to vertices in the interior of~$C$.  We
call non-separating 4-cycles also \emph{empty 4-cycles}.

Consider an $n$-vertex PTP graph $G$ where all separating 4-cycles are
trivial.  Then any path in $\cH(G)$ has length at most
$\Oh(n^2)$~\cite{EMSV12}.
To extend this result~to general PTP graphs, we use a recursive
decomposition of $G$ following Eppstein \etal~\cite{EMSV12} (as well
as Dasgupta and Sur-Kolay~\cite{DSK01} and Mumford~\cite{Mum08}).
For a nontrivial separating 4-cycle $C$ of a PTP graph $G$, we define
two \emph{separation components} of $G$ with respect to $C$: The
\emph{inner separation component} $\GCi$ is the maximal subgraph of
$G$ with $C$ as outer face.  The \emph{outer separation component}
$\GCo$ is the minor of $G$ obtained by replacing the interior of $C$
with a single vertex~$v_C$.  Note that both $\GCi$ and $\GCo$ are PTP
graphs.  A \emph{minimal separation component} of $G$ is a separation
component of $G$ that cannot be split any further.  Partitioning~$G$
into minimal separation components takes linear time~\cite{EMSV12}.

Eppstein \etal~\cite{EMSV12} pointed out that there can be a quadratic
number of nontrivial separating 4-cycles in $G$, but that these can be
represented in linear space by finding all maximal complete bipartite
subgraphs $K_{2,i}$ of $G$.  Such a representation can be found in
linear time~\cite{EMSV12}.
Note that, among the 4-cycles of a subgraph $K_{2,i}$ (with $u$ and
$v$ forming the small partition), only the outermost 4-cycle $C$ is
rotatable.  This can be seen as $C$ is the outer 4-cycle of a smaller
PTP graph and thus any of the smaller cycles has a monochromatic path
from~$u$ to~$v$; see~\cref{fig:cycleRelationsApp:Kmn,fig:cycleRelationsApp:intersection}.

\begin{observation}
  \label[observation]{clm:cycleRelations}
  Let $G$ be a PTP graph with two rotatable 4-cycles~$C$ and~$C'$.
  Then the interiors of $C$ and $C'$ are either disjoint or one lies
  inside the other.
\end{observation}

\begin{proof}
  We prove that it is not possible that the interiors of $C$ and $C'$
  intersect properly.  Assume otherwise and observe that then $C$ and
  $C'$ must intersect in two non-adjacent vertices~$u$ and $v$ such
  that each contains one vertex of the other; see
  \cref{fig:cycleRelationsApp:intersection}.  However, then $u$ and $v$ form one
  partition of a subgraph $K_{2,i}$.  Since neither $C$ nor $C'$ is
  the outermost 4-cycle of this $K_{2,i}$ neither can be~rotatable.
\end{proof}

\begin{figure}[tb]
  \centering
  	\centering
	    \begin{subfigure}[t]{.47 \linewidth}
		\centering
		\includegraphics[page=1]{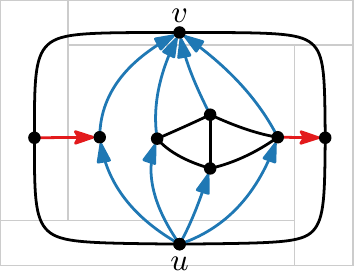}
		\caption{A PTP graph with a $K_{2,5}$ subgraph.}
		\label{fig:cycleRelationsApp:Kmn}
	\end{subfigure}
	\hfill
	\begin{subfigure}[t]{.47 \linewidth}
		\centering
		\includegraphics[page=2]{figs/cycleRelations}
		\caption{Two properly intersecting 4-cycles.}
		\label{fig:cycleRelationsApp:intersection}
	\end{subfigure}
  \caption{(a) Only the outermost 4-cycle of this $K_{2,5}$
    subgraph of a PTP graph can be rotatable; (b)~hence, two 4-cycles
    $C$ and $C'$ with properly intersecting interiors cannot be
    rotatable.}
  \label{fig:cycleRelationsApp}
\end{figure}
\begin{figure}[tb]
  \centering
  \includegraphics[page=3]{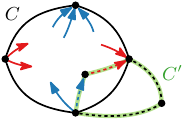}
  \caption{Two rotatable 4-cycles $C$ and $C'$ may overlap if, say, $C'$ is empty and its
    interior contains exactly one edge of~$C$.}
  \label{fig:cycleRelationsOverlapApp}
\end{figure}

Note that two rotatable 4-cycles $C$ and $C'$ may \emph{overlap} in
the sense that one of them, say, $C'$ is empty and its interior
contains exactly one edge of $C$; see \cref{fig:cycleRelationsOverlapApp}.
Indeed, if $C$ rotates multiple times on a path from
$\cL_{\mathrm{min}}$ to $\cL_{\mathrm{max}}$, then each edge of $C$ is
interior to an empty rotatable 4-cycle.

\Cref{clm:cycleRelations} implies that a rotation of a 4-cycle $C'$
for a REL $\cL$ of $G$ can be seen as a rotation on $\cL$ restricted
to $\GCi$ if $C'$ lies inside of $C$, or restricted to $\GCo$ if their
interiors are disjoint or if~$C$ lies inside of~$C'$.  In the last
case, to combine the RELs of $\GCo$ and $\GCi$ into a REL of~$G$, we
have to rotate the REL of $\GCi$ once.  This yields the following.

\begin{lemma}
  \label{clm:rotationSequencePartition}
  Let $G$ be a PTP-graph with a nontrivial rotatable separating
  4-cycle~$C$.  A rotation sequence $\sigma$ on a REL $\cL$ of $G$ can
  be partitioned into a rotation sequence that acts on $\cL$
  restricted to $\GCo$ and a rotation sequence that acts on $\cL$
  restricted to $\GCi$.
\end{lemma}

\RELdistance*
\label{clm:RELdistance*}

\begin{proof}
For a PTP graph that contains only empty and trivial separating 4-cycles
or, equivalently, for a minimal separation component, this is known~\cite{EMSV12}.
We partition~$G$ into two separation components and 
apply an inductive argument based on~\cref{clm:rotationSequencePartition}.
Since the bound on the diameter is on the number of inner vertices
(that is, we do not count the outer 4-cycle), the sum of the
lengths of the sequences in the two separation components is still in~$\Oh(n^2)$.
\end{proof}

\paragraph{Construction of a shortest path.}
\label{app:shortestPath}

Let~$G$ be an~$n$-vertex PTP graph whose separating 4-cycles are trivial.
Let~$\R$ and~$\R'$ be two rectangular duals of~$G$ that realize the RELs~$\cL$ and~$\cL'$, respectively.
For a rotatable 4-cycle of~$G$, we defined the rotation count~$f_C(\cL)$
as the number of rotations of~$C$ in a monotone sequence from the minimum REL~$\cL_{\mathrm{min}}$ to~$\cL$. 
Let~$\cP(G)$ be a poset where each vertex is a tuple~$(C, i)$ 
consisting of a rotatable 4-cycle~$C$ and a rotation count~$i$.
Furthermore,~$(C, i) \leq (C', j)$ when for each REL~$\cL$ with~$f_C(\cL) \leq i$
it holds that~$f_{C'}(\cL) \leq j$;
that is, it is not possible to increase the rotation count of~$C'$ from~$j$ to~$j+1$
prior to increasing the rotation count of~$C$ from~$i$ to~$i+1$.
Eppstein \etal~\cite{EMSV12} showed that each REL~$\cL$ of~$G$ can be represented
by the partition of~$\cP(G)$ into a downward-closed set~$L(\cL)$ 
and an upward-closed set~$U(\cL)$ such that
$(C, i) \in L(\cL)$ when~$i < f_C(\cL)$ and~$(C, i) \in U(\cL)$ otherwise.
In fact, they showed that~$\cP(G)$ is order-isomorphic 
to the poset defined by Birkhoff's representation theorem~\cite{Bir37}.
Hence, the meet of two RELs~$\cL$ and~$\cL'$ in~$\cH(G)$ (the lattice of RELs of~$G$) 
is represented by~$L(\cL) \cap L(\cL')$.
By \cref{clm:cycleRelations,clm:RELdistance}, 
we can extend this result again to any PTP graph~$G$ 
by considering the minimum separation components of~$G$ separately.
The rotation count of each rotatable four-cycle
is then define only with respect to the minimum separation component of~$G$ it lies in.

Let $C$ be a rotatable 4-cycle of $G$.
We define the \emph{rotation count} $f_C(\cL)$ as the number of rotations of $C$ 
in a monotone sequence from the minimum REL $\cL_{\mathrm{min}}$ to $\cL$. 
Let $\bar \cL = \cL \wedge \cL'$, the meet of $\cL$ and $\cL'$.
By a result of Eppstein \etal~\cite{EMSV12},
for every rotatable 4-cycle $C$ of $G$,
it holds that $f_C(\bar\cL) = \min\set{f_C(\cL), f_C(\cL')}$.
Details are given in \cref{app:shortestPath}.

\RELshortestPath*
\label{clm:RELshortestPath*}
\begin{proof}
Let $\bar \cL = \cL \wedge \cL'$ be the meet of $\cL$ and $\cL'$.
As noted in \cref{app:lattice}, there exists a shortest path from $\cL$ to $\cL'$ via $\bar \cL$.
Therefore, we first compute $\bar \cL$ and then construct monotone paths from $\cL$ and $\cL'$ to $\bar \cL$, respectively.
To this end, we first find all cw rotatable 4-cycles in $\cL$
and then execute rotations until we reach $\cL_{\mathrm{min}}$.
With each step, we count the number of times the involved cycles have been rotated so far
and check whether this enables cw rotations of other 4-cycles.
When we reach $\cL_{\mathrm{min}}$, we have the rotation counts for all (involved) rotatable 4-cycles for $\cL$ (and $\cL'$).
For each rotatable 4-cycle $C$ of $G$,
we compute $f_C(\bar\cL)$ as $\min\set{f_C(\cL), f_C(\cL')}$.
One such rotation step takes $\Oh(n)$ time and by \cref{clm:RELdistance} the number of steps is at most $\Oh(n^2)$.

Finally, to compute two shortest paths from $\cL$ and~$\cL'$ to $\bar \cL$, 
we greedily rotate cw rotatable 4-cycles whose rotation count is great than~$f_C(\bar \cL)$.  
This takes~$\Oh(n^2)$~time.
\end{proof}

In \cref{sec:serial}, we construct a relaxed morph between two rectangular duals $\R$ and $\R'$ along a shortest paths between them.
This morph requires at most a quadratic number of steps,
since a shortest path between $\R$ and $\R'$  has by \cref{clm:RELshortestPath} at most quadratic length.
Note that for this asymptotic bound on the number of steps,
it would suffice to morph along, for example, a path from $\R$ 
via a rectangular dual that realizes $\cL_{\mathrm{min}}$ to $\R'$. 
By \cref{clm:cycleRelations}, this path also has a length in $\Oh(n^2)$.
However, with \cref{clm:approx}, we show that by executing rotations in parallel
and morphing along a shortest path, we can construct a morph that uses only $\Oh(1)$ times 
the minimum number of steps needed to get from~\R to~$\R'$.

\section{Using Kant and He on Auxiliary Graph \texorpdfstring{$\hat G$}{Ĝ}} \label{app:kanthe}
In this section we explain the algorithm by Kant and He~\cite{KH97} 
and explain why we can use it on the auxiliary graph $\hat G$ with auxiliary REL $\hat \cL$ 
in the proof of \cref{clm:tidyRecDual}.

Kant and He~\cite{KH97} introduced RELs
and described two linear-time algorithms that compute a REL for a given PTP graph;
one algorithm is based on edge contractions, the other is based on canonical orderings.
They then use the algorithm by He~\cite{He93} to construct in linear time a 
rectangular dual that realizes this REL
and where the coordinates are all integers.
 
The algorithm by He works as follows; see \cref{fig:weakDual}.
Given a PTP graph $G$ with REL $(L_1,L_2)$.
Consider the weak dual $L_1^*(G)$ of $L_1(G)$. 
In $L_1^*(G)$, there is a vertex $f$ for every interior face $f$ in the plane embedding of $L_1(G)$, 
and there are two vertices $\overline{f}_{\mathrm{L}}$ and 
$\overline{f}_{\mathrm{R}}$ for the outer face $\overline{f}$ in the plane embedding of $L_1(G)$.
There is an edge $(f,f')$ if $f$ is the interior face to the left and $f'$ is 
the interior face to the right of some edge in $L_1(G)$;
an edge $(\overline{f}_{\mathrm{L}},f')$ if $\overline{f}$ is the face to the left and $f'$ is 
the interior face to the right of some edge in $L_1(G)$;
and an edge $(f,\overline{f}_{\mathrm{R}})$ if $f$ is the interior face to the left and $\overline{f}$ is 
the face to the right of some edge in $L_1(G)$.
Then $L_1^*(G)$ is a planar st-graph with source $\overline{f}_{\mathrm{L}}$
and sink $\overline{f}_{\mathrm{R}}$.
Compute a topological numbering $d$ of $L_1^*(G)$.
For any vertex $v$ of $G$, let $\mathrm{left}(v)$ be the face to the left of $v$
in $L_1(G)$, and let $\mathrm{right}(v)$ be the face to the right of $v$
in $L_1(G)$; these are the two faces between an incoming and an outgoing edge
of $v$. Then the algorithm sets the x-coordinate of the left side of $\R(v)$
to $d(\mathrm{left}(v))$ and the x-coordinate of the right side of $\R(v)$
to $d(\mathrm{right}(v))$.
The y-coordinates for the bottom and top side of $\R(v)$ are calculated
analogously from the weak dual $L_2^*(G)$ of $L_2(G)$.

\begin{figure}[ht]
	\centering
	    \begin{subfigure}[t]{.34 \linewidth}
		\centering
		\includegraphics[page=1]{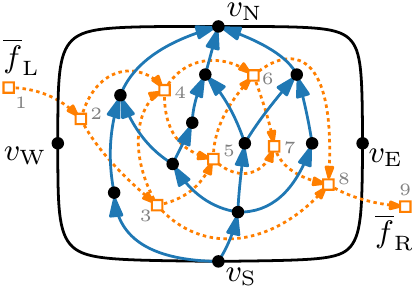}
		\caption{$L_1(G)$ and $L_1^*(G)$}
		\label{fig:weakDual:one}
	\end{subfigure}
	\hfill
	\begin{subfigure}[t]{.33 \linewidth}
		\centering
		\includegraphics[page=2]{weakDual}
		\caption{$L_2(G)$ and $L_2^*(G)$}
		\label{fig:weakDual:two}
	\end{subfigure}
	\hfill
	\begin{subfigure}[t]{.3 \linewidth}
		\centering
		\includegraphics[page=3]{weakDual}
		\caption{$\R'$}
		\label{fig:weakDual:R}
	\end{subfigure}
	\caption{Illustration of the algorithm by He~\cite{He93}.
          (a+b) Topological numberings (gray) for the weak duals
          $L_1^*(G)$ and $L_2^*(G)$ (dotted) of $L_1(G)$ and $L_2(G)$;
          (c)~resulting rectangular dual~\R.}
	\label{fig:weakDual}
\end{figure}

Note that the maximal vertical segments in the resulting rectangular
dual~\R are in bijection with the vertices of~$L_1^*(G)$ and the
maximal horizontal segments are in bijection with the vertices
of~$L_2^*(G)$; see \cref{fig:weakDual:R}. 

He~\cite{He93} did not show this explicitly, but for a given REL, his
linear-time algorithm yields a rectangular dual of minimum width and
height, and thus of minimum area and perimeter. 

In the proof of \cref{clm:tidyRecDual} we extended a PTP graph $G$ 
to an auxiliary graph $\hat G$ where some faces are chordless 4-cycles. 
We also described how the REL $\cL$ of $G$ 
can be extended to an auxiliary REL $\hat L$ of $\hat G$
that adheres to the REL coloring rules.
Also, each chordless 4-cycle $C$ bounding a face of $\hat G$ is colored alternatingly.
In a rectangular dual $\hat \R_1$ of $\hat G$, 
the four rectangles corresponding to the vertices of $C$ meet in a single point;
there are thus not only T-junctions, but also crossings.
The vertices of $\hat L_1^*(G)$ and $\hat L_2^*(G)$ 
thus still correspond to maximal vertical and horizontal line segments in $\hat \R_1$, respectively.
See also again \cref{fig:prepMorph,fig:prepMorph:empty}. 
Hence, applying the algorithm by Kant and He~\cite{KH97} to~$\hat{G}$
and~$\hat{\cL}$, we obtain the (almost) rectangular dual~$\hat{\R_1}$
described above.

\section{Proofs Omitted in Section \ref{sec:parallel}}

\label{clm:simultaneousRotation*}
\simultaneousRotation*

\begin{proof}
We assume that $\cC$ contains only cw rotatable cycles;
otherwise we apply the same algorithm to the ccw rotatable cycles afterwards. 
Our algorithm to construct a morph from $\R$ to some $\R'$ consists of the following steps.
(i) We construct the conflict graph~$K(\cC)$.
This can be done in linear time by traversing 
the four border segments of each cycle in $\cC$. 
(ii) We greedily compute a 5-coloring of $K(\cC)$, again in linear time,
and pick a set $\cC_i$, $i \in \set{1, \ldots, 5}$, corresponding to one color.
We add all empty 4-cycles of $\cC$ to $\cC_1$.
(iii) We compute and execute a preparatory linear morph such that,
(iv) all rotations in $\cC_i$ can safely be executed in parallel.
Hence, we need a constant number of linear morphs 
to arrive at a rectangular dual $\R'$ realizing $\cL'$.

The preparatory linear morph for $\cC_i$ works similar to the preparatory morph in the serial case.
Using an auxiliary PTP graph $\hat G$ and an auxiliary REL $\hat \cL$
we find a target rectangular dual $\R_1$ that also realizes $\cL$.
We can then use a linear morph from $\R$ to $\R_1$ according to \cref{sec:sameREL}.
For each 4-cycle in $\cC_i$ that does not share a segment with another 4-cycle in $\cC_i$,
we extend $G$ and $\cL$ towards $\hat G$ and $\hat \cL$, respectively, as before.
Since these cycles are non-conflicting, this can be done nearly independently for each 4-cycle;
however a rectangle might now be split both horizontally and vertically (meaning, its vertex gets duplicated twice). 

For a separating 4-cycle $C$ and empty 4-cycles $C_1, \ldots, C_k$ 
that share a maximal line segment $s$ as border segment,
we assume w.l.o.g.\ that $s$ is the upper border segment of $C$. 
Further, let $C_1, \ldots, C_k$ appear along $s$ from left to right
and let $C = \croc{a, b, c d}$ as in the proof of \cref{clm:tidyRecDual}.
This case is illustrated in \cref{fig:simultaneousRotationSegment:source}. 
The rotation of $C$ moves $s$ downwards
Note that the rotation of $C$ breaks $s$ and moves the part right of the interior $I_C$ of $C$ down.
Thus, each 4-cycle $C_i$, $i \in \set{1, \ldots, k}$
has to follow this downward motion.
Furthermore, $C_i$ move the part of $s$ right of it also one down
and hence $C_i$ moves down the height of $I_C$ plus $i$; see \cref{fig:simultaneousRotationSegment:target}.
Therefore, to create enough vertical space below $s$ that contains not other horizontal segment,
we ensure that in $\R_1$ each rectangle $\R_1(v)$ with $v$ in $P_d \setminus \set{x}$
satisfies $\y_1(\R_1(v)) < \y_2(\R_1(a)) - k$; see \cref{fig:simultaneousRotationSegment:auxiliary}.
Recall that the preparatory step for $C$ in the serial case
duplicates all vertices on $P_d$.
We do the same here with a minor change,
namely, we split $d$ into the vertices $d_1, \ldots, d_{k+2}$
while all other vertices in $P_d \setminus \set{d, x}$ are only duplicated.
Let $y$ be the successor of $d$ on $P_d$.
We assign the edges clockwise between (and including) $dy$ and $ad$ to $d_1$, 
and the edges clockwise between $ad$ and $dy$ to $d_{k+2}$.
For $i \in \set{1, \ldots, k + 1}$, we connect $d_i$ and $d_{i+1}$ with a blue edge and
add a red edge $a d_i$.
Furthermore, we add the red edge $d_1 y_1$ and, for $\in \set{2, \ldots, k + 2}$,
the red edge $d_i y_2$.
For the vertices $v \in P_d \setminus \set{d}$, we add edges as before.
As a result we get that $\hat\R_1(d_1), \ldots, \hat\R_1(d_k)$ stack on top of each other 
and, for all $v$ in $P_d \setminus \set{x}$, we have $\y_1(\hat\R_1(v)) < \y_2(\hat\R_1(a)) - k$.
In other words, there is enough vertical space below $s$ in $\hat \R_1$ and thus also in $\R_1$.

The case where $s$ is the interior segment shared only by empty 4-cycles
is handled analogously to the previous case and \cref{clm:tidyRecDual}.

\begin{figure}[htb]
  \centering
    \begin{subfigure}[t]{.31 \linewidth}
		\centering
		\includegraphics[page=1]{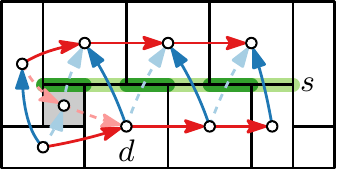}
		\caption{In $\R$ the separating 4-cycle and the two empty 4-cycles share $s$ as border segment.}
		\label{fig:simultaneousRotationSegment:source}
	\end{subfigure}
	\hfill
	\begin{subfigure}[t]{.31 \linewidth}
		\centering
		\includegraphics[page=2]{simultaneousRotations}
		\caption{We create enough vertical space for (parts of) $s$ to move downwards 
			by splitting $d$ sufficiently many times in $\hat\R_1$.}
		\label{fig:simultaneousRotationSegment:auxiliary}
	\end{subfigure}
	\hfill
	\begin{subfigure}[t]{.31 \linewidth}
		\centering
		\includegraphics[page=3]{simultaneousRotations}
		\caption{In $\R'$ each successive part of $s$ moves further down.}
		\label{fig:simultaneousRotationSegment:target}
	\end{subfigure}
  \caption{To enable a parallel rotation of all three alternating 4-cycles 
  in the rectangular dual $\R$ that share the border segment $s$,
  we first morph to a rectangular dual $\R_1$ derived from the auxiliary rectangular dual $\hat\R_1$.
  In doing so, we ensure that there is enough space below $s$ such that its parts can move down 
  to their positions in $\R'$.}
  \label{fig:simultaneousRotationSegment} 
\end{figure}

The auxiliary graph $\hat G$ and the auxiliary REL $\hat \cL$ 
can be constructed locally around each involved 4-cycle.
Furthermore, since each vertex is either only duplicated a constant number of times
or one vertex linear many times in the number of involved empty 4-cycles (such as $d$ above),
the sizes of $\hat G$ and $\hat \cL$ are linear in the size of $G$. 
Hence, the algorithm by Kant and He~\cite{KH97} computes $\hat \R_1$ in $\Oh(n)$ time
and we can derive $\R_1$ in $\Oh(n)$ time.
Finally, to find $\R'$, we compute the positions of all rectangles after the rotations. 
Because of the preparatory step, this can be again done locally on $\R_1$ in overall $\Oh(n)$ time.
\end{proof}

\label{clm:approx*}
\approx*
\begin{proof}
Let~$\R$ and~$\R'$ be realized by RELs~$\cL$ and~$\cL'$, respectively.
Let~$\bar \cL = \cL \wedge \cL'$, the meet of~$\cL$ and~$\cL'$.
We first consider the path from~$\cL$ to~$\bar \cL$. 
Let~$X$ be the set of cw rotations between~$\cL$ and~$\bar \cL$. 
Not all cw rotations in~$X$ can be executed in parallel 
because other cw rotations in~$X$ might have to be executed first.
Let~$k$ be the minimum number of linear morphs of a relaxed morph~$M$
from~$\cL$ to~$\bar \cL$. This morph has to execute all rotations in $X$, so
it partitions~$X$ into sets~$X_1, \ldots, X_k$ 
where~$X_i$ contains all cw rotations in $X$ executed by~$M$ in step~$i$.
Our algorithm repeatedly applies \cref{clm:simultaneousRotation} to
execute all possible cw rotations in $X$ within $\Oh(1)$ steps.
Hence, in the first application of \cref{clm:simultaneousRotation}, our algorithm
executes all rotations in $X_1$ (and possibly more).
Subsequently, in the second application of \cref{clm:simultaneousRotation}, our
algorithm executes all rotations in $X_2$ (except those it had already
executed in the first application).
Let $x$ be a rotation that is executed in the $j$-th application of
\cref{clm:simultaneousRotation}.
Then we have $x\in X_i$ with $i\ge j$.
Thus, our algorithm requires at most $k$ applications of \cref{clm:simultaneousRotation}
to execute all rotations in $X=X_1 \cup\dots\cup X_k$,
each of which requires $\Oh(1)$ steps.  In total, our algorithm
uses~$\Oh(k)$ steps to get from~$\cL$ to~$\bar \cL$.

Symmetrically, let~$X'$ be the set of ccw rotations between~$\bar\cL$ and~$\cL'$.
If~$\ell$ is the minimum number of linear morphs 
of a relaxed morph from~$\bar \cL$ to~$\cL'$,
then our algorithm requires~$\Oh(\ell)$ steps between~$\bar \cL$ and~$\cL'$,
so in total~$\Oh(k + \ell)$ steps from~$\cL$ to~$\cL'$. 

Let~$M^\star$ be a relaxed morph between~$\R$ and~$\R'$ 
that uses the minimum number of linear morphs.
Because of the lattice structure,~$M^\star$ has to execute all cw rotations in~$X$ 
and all ccw rotations in~$X'$.
Hence, $M^\star$ uses at least~$\max\set{k, \ell} \geq (k +
\ell)/2$ linear morphs.

The running time follows from \cref{clm:serialMorphing,clm:simultaneousRotation}.
\end{proof}

\end{document}

%% file: config.tex
 	\usepackage[T1]{fontenc}
\usepackage{amsmath,amssymb}
\usepackage{xspace}
\usepackage[usenames,dvipsnames,table]{xcolor}
\usepackage{wrapfig,graphicx}
\DeclareGraphicsExtensions{.pdf,.png,.eps}
\graphicspath{{figs/}{notes/}}
\usepackage[inline]{enumitem}
\newcommand{\appmark}{$\star$}
\usepackage[font=small, labelfont=bf]{caption}
\usepackage[font=small, labelfont=normal]{subcaption}

\definecolor{darkblue}{rgb}{0.121,0.47,0.705}
\newcommand{\bl}{\color{darkblue}}
\let\emph\relax
\DeclareTextFontCommand{\emph}{\bl\em}

\usepackage{mathtools}
\DeclarePairedDelimiter\set{\{}{\}}

\DeclarePairedDelimiter\croc{\langle}{\rangle}
\def\Oh{\ensuremath{\mathcal{O}}}

\def\R{\ensuremath{\mathcal{R}}\xspace}

\def\cL{\ensuremath{\mathcal{L}}\xspace}
\def\cC{\ensuremath{\mathcal{C}}\xspace}
\def\cH{\ensuremath{\mathcal{H}}\xspace}
\def\cP{\ensuremath{\mathcal{P}}\xspace}

\def\vWest{\ensuremath{v_\mathrm{W}}\xspace}
\def\vSouth{\ensuremath{v_\mathrm{S}}\xspace}
\def\vEast{\ensuremath{v_\mathrm{E}}\xspace}
\def\vNorth{\ensuremath{v_\mathrm{N}}\xspace}

\def\GCi{\ensuremath{G_C^\mathrm{in}}\xspace}
\def\GCo{\ensuremath{G_C^\mathrm{out}}\xspace}
\DeclareMathOperator{\y}{y}
\newtheorem{observation}[theorem]{Observation}
\let\doendproof\endproof
\renewcommand\endproof{~\hfill$\qed$\doendproof}

\usepackage{thmtools} 
\usepackage{thm-restate}
\usepackage{url}
\usepackage{cite}
\usepackage[hidelinks=true,
			bookmarks=true,
			bookmarksopen=true,
			bookmarksopenlevel=2,
			bookmarksnumbered=true,
			hyperindex=true,
			plainpages=false,
			pdfpagelabels=true
]{hyperref} 
\usepackage[capitalise,nameinlink]{cleveref}
\crefname{observation}{Obs.}{Obs.}
\Crefname{observation}{Observation}{Observations}
\crefname{proposition}{Prop.}{Props.}
\Crefname{proposition}{Proposition}{Propositions}
\crefname{section}{Sect.}{Sects.}
\Crefname{section}{Section}{Sections}
\newcommand{\etal}{{et~al.}\xspace}

\renewcommand{\orcidID}[1]{\href{https://orcid.org/#1}{\includegraphics[scale=.03]{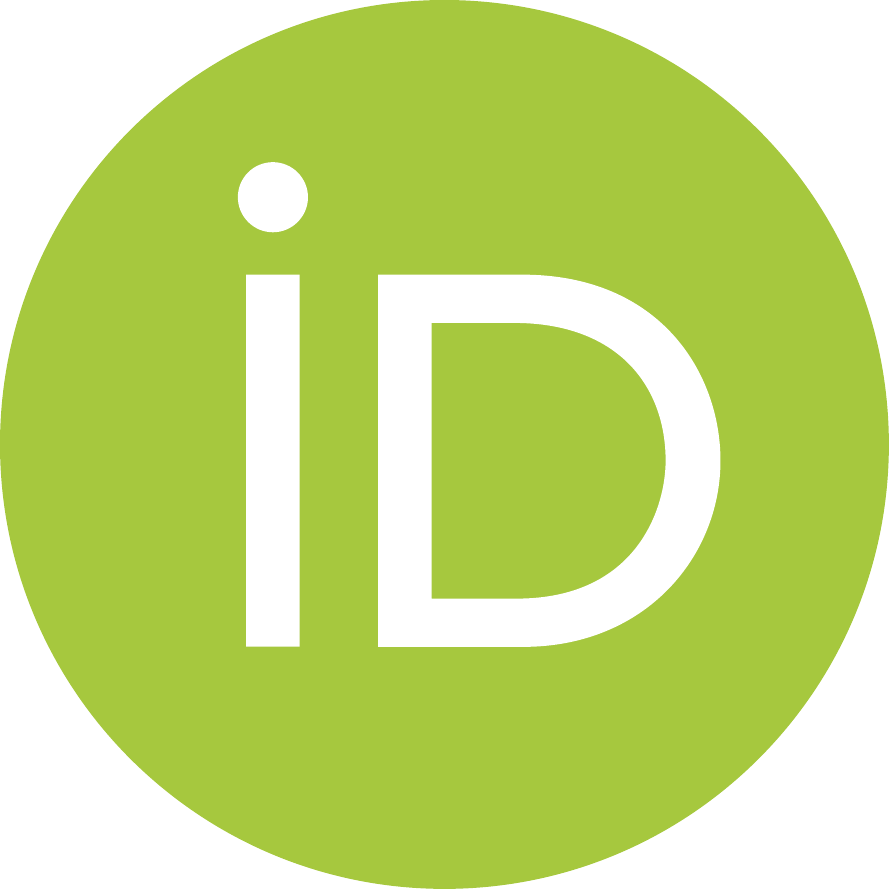}}}
